\documentclass[twocolumn,secnumarabic,amssymb, nobibnotes, aps, pra, 10pt]{revtex4-2}

\usepackage{hyperref}
\usepackage{amsthm}
\usepackage{amsmath}
\usepackage{amssymb}
\usepackage{graphicx}
\usepackage[draft]{changes}

\allowdisplaybreaks

\theoremstyle{definition}
  \newtheorem{defn}{\protect\definitionname}
\theoremstyle{plain}
\newtheorem{thm}{\protect\theoremname}
  \theoremstyle{plain}
  
  \theoremstyle{plain}
  
 \theoremstyle{definition}

  \providecommand{\definitionname}{Definition}
  \providecommand{\examplename}{Example}
  \providecommand{\propositionname}{Proposition}
\providecommand{\corollaryname}{Corollary}
\providecommand{\theoremname}{Theorem}

\begin{document}

\title{Entanglement measures for two-particle quantum histories}%

\author{Danko Georgiev}

\affiliation{Institute for Advanced Study, 30 Vasilaki Papadopulu Str., Varna 9010, Bulgaria}

\email{danko.georgiev@mail.bg}

\author{Eliahu Cohen}

\affiliation{Faculty of Engineering and the Institute of Nanotechnology and Advanced Materials, Bar Ilan University,
Ramat Gan 5290002, Israel}

\email{eliahu.cohen@biu.ac.il}


\date{December 10, 2022}

\begin{abstract}
Quantum entanglement is a key resource, which grants quantum systems the ability to accomplish tasks that are classically impossible. Here, we apply Feynman's sum-over-histories formalism to interacting bipartite quantum systems and introduce entanglement measures for bipartite quantum histories. Based on the Schmidt decomposition of the matrix comprised of the Feynman propagator complex coefficients, we prove that bipartite quantum histories are entangled if and only if the Schmidt rank of this matrix is larger than 1. The proposed approach highlights the utility of using a separable basis for constructing the bipartite quantum histories and allows for quantification of their entanglement from the complete set of experimentally measured sequential weak values. We then illustrate the non-classical nature of entangled histories with the use of Hardy's overlapping interferometers and explain why local hidden variable theories are unable to correctly reproduce all observable quantum outcomes. Our theoretical results elucidate how the composite tensor product structure of multipartite quantum systems is naturally extended across time and clarify the difference between quantum histories viewed as projection operators in the history Hilbert space or viewed as chain operators and propagators in the standard Hilbert space.
\end{abstract}

\maketitle

\section{Introduction}

Quantum entanglement is a remarkable physical property of composite quantum systems,
first conceptualized by Erwin Schr\"{o}dinger in 1935,
such that it allows for the realization of physical situations in
which the best possible knowledge of the composite system does not
include the best possible knowledge of its constituent parts \cite{Schrodinger1935,Bengtsson2017,Brukner2021}.
Although the possibility of quantum entanglement was initially met with skepticism \cite{Einstein1935}, it is now well-confirmed by experiments \cite{Aspect1982,Aspect1999,Vedral2014} and is utilized routinely in cutting-edge quantum applications including quantum computing \cite{Nielsen2010,Bouwmeester2000},
quantum teleportation \cite{Bennett1993,Bouwmeester1997}, quantum
superdense coding \cite{Bennett1992}, quantum erasure \cite{Herzog1995},
quantum random number generation \cite{Jennewein2000} and quantum
cryptography \cite{Shenoy2017}. Because these quantum tasks cannot
be accomplished by classical systems, quantum entanglement is regarded
as a valuable physical resource that is worth producing and whose practical
value can be precisely quantified \cite{Horodecki2009,Wootters1998,Chitambar2019}.

For pure bipartite quantum states $|\Psi(t)\rangle_{AB}$ given at a single
time $t$ in a composite Hilbert space $\mathcal{H}=\mathcal{H}_A\otimes\mathcal{H}_B$, quantification of quantum entanglement is achieved
by the Schmidt decomposition of the state \cite{Schmidt1907a,Stewart1993,Ekert1995,Miszczak2011},
which provides a set of non-negative singular values $\lambda_{1}\geq\lambda_{2}\geq\ldots\geq\lambda_{s}\geq0$
referred to as Schmidt coefficients that can be then used for the
construction of various entanglement measures \cite{Vedral1997,Vedral1998,Horodecki2001,Plenio2007,Gudder2020b,Androulakis2020,Georgiev2022}.
In the presence of quantum dynamics, however, the quantum entanglement
in the system can change in time due to the quantum interaction Hamiltonian
that governs the physical interactions between the component subsystems
\cite{Georgiev2021b}. As a result, the individual quantum histories
for the component subsystems may entangle and it is not
immediately clear how this entanglement between the quantum histories
can be quantified. In fact, in many experimental setups it is customary
to prepare a separable initial quantum state, $|\Psi(0)\rangle_{AB}=|\Psi(0)\rangle_{A}\otimes|\Psi(0)\rangle_{B}$,
which is then dynamically evolved and eventually subjected to a final quantum measurement.
Because individual quantum histories may start or end with separable states, and the quantum entanglement may also dynamically evolve at intermediate times, it is desirable to have a general procedure that quantifies the overall entanglement of quantum histories.
Our work on the problem is motivated by a growing interest in the concept, including the recent theoretical development of the ``entangled histories'' formalism, which focuses on time factorization and allows for experimental tests of entangled quantum histories in time \cite{Cotler2016,Cotler2017b,Dong2017,Nowakowski2018,Pan2019}.
The entangled histories formalism was developed to enable reconstruction of the past evolution of a quantum system from measurements in the present \cite{Cotler2016,Cotler2017b}. The resulting quantum entangled histories exhibit non-classical features such as superposition of time evolutions and violation of Bell-type inequalities, indicating that the measurement outcomes produced by single quantum systems are not predetermined, but contextual on the entire set of commuting quantum observables that are simultaneously measured \cite{Dong2017,Nowakowski2018,Pan2019}. Entanglement in time alone, however, does not grant single quantum systems the capacity to produce nonlocal correlations between space-like separated regions. Instead, one must allow for quantum interactions between multiple subsystems and exploit the resulting spatial entanglement. In this work, we extend the entangled quantum histories formalism to multiple interacting quantum subsystems and provide quantitative entanglement measures for the bipartite case.

For quantification of entanglement in a complete set of bipartite quantum histories, which start from a preselected initial state and
end with a postselected final state, we have applied
Feynman's sum-over-histories formalism to describe the dynamics of the composite system.
Then, we have generalized the standard Schmidt decomposition
of quantum states at a single time to Schmidt decomposition of complete
sets of Feynman propagators for multi-time quantum histories and
used the resulting Schmidt coefficients for the calculation of entanglement
measures of these histories. A convenient feature
of this approach is that it avoids the use of entangled projectors
in the construction of the quantum history projection operators in
history Hilbert space and does not require the computation of inner
products between pairs of quantum history chain operators in standard Hilbert space.\\

The organization of the subsequent sections is as follows:
In Section~\ref{sec:2}, we introduce the main definitions for quantum
history projective operators in the history Hilbert space, with corresponding
chain operators and Feynman propagators in the standard
Hilbert space. Then, we elaborate on the construction of Feynman propagator
complex coefficient matrix from which can be determined the Schmidt coefficients
for the calculation of different entanglement measures including the
entanglement entropy, entanglement robustness and concurrence.
The possibility of experimental measurement of sequential weak values, instead of Feynman propagators, is also considered.
In Section~\ref{sec:3}, we describe Hardy's interferometer and use it as a test bed
in Section~\ref{sec:4} to highlight the differences between classical histories and quantum histories.
In Section~\ref{sec:5}, we provide a graphical illustration of entangled quantum histories and
explicitly compute their concurrence.
Lastly, in Section~\ref{sec:6},
we show how entangled quantum histories can be used for experimental
demonstration of quantum nonlocality for spacelike separated quantum
measurement outcomes. We conclude this work with a brief discussion
on the utility of the quantum history formalism for the description of optical applications.
Prospects for future research are also presented.

\section{\label{sec:2}Entanglement of quantum histories}

For the analysis of entangled quantum histories, we will need to introduce
three types of related, but mathematically distinct objects: quantum history projection operators in history Hilbert space, as well as
quantum history chain operators and corresponding Feynman propagators in the standard Hilbert space.

\begin{defn}
(Quantum history projection operators) Given a quantum system residing in $n$-dimensional Hilbert space $\mathcal{H}$, quantum histories are constructed with the use of $k$ temporal resolutions of the identity operator $\hat{I}$ given
by complete sets of projection operators

$\{\mathcal{\hat{P}}_{1}^{(t_{1})},\hat{\mathcal{P}}_{2}^{(t_{1})},\ldots,\mathcal{\hat{P}}_{n}^{(t_{1})}\}$,
$\{\mathcal{\hat{P}}_{1}^{(t_{2})},\hat{\mathcal{P}}_{2}^{(t_{2})},\ldots,\mathcal{\hat{P}}_{n}^{(t_{2})}\}$,
$\ldots$, $\{\mathcal{\hat{P}}_{1}^{(t_{k})},\hat{\mathcal{P}}_{2}^{(t_{k})},\ldots,\mathcal{\hat{P}}_{n}^{(t_{k})}\}$
that span the Hilbert space of the system at each time point $t_{1},t_{2},\ldots,t_{k}$.
The superscripts denoting different times are written in brackets in order to remind that these are not powers and to indicate that projectors with the same subscript but with different superscript do not have to be equal, namely, the resolution of identity can be performed in different bases at different time points.
Next, we associate $k$ distinct indices each enumerating the $n$ projectors at the corresponding time point, namely,
$i_{1}\in\{1,2,\ldots,n\},i_{2}\in\{1,2,\ldots,n\},\ldots,i_{k}\in\{1,2,\ldots,n\}$, so that we can handle simultaneously $k$~sums:
$\sum_{i_{1}=1}^{n}\mathcal{\hat{P}}_{i_{1}}^{(t_{1})}=\hat{I}_{1}$, 
$\sum_{i_{2}=1}^{n}\mathcal{\hat{P}}_{i_{2}}^{(t_{2})}=\hat{I}_{2}$,
$\ldots$, $\sum_{i_{k}=1}^{n}\mathcal{\hat{P}}_{i_{k}}^{(t_{k})}=\hat{I}_{k}$.
If we employ the symbol $\odot$ to denote tensor products at different times,
we can write each \emph{quantum history} as a one-dimensional projection
operator 
\begin{equation}
\hat{\mathcal{Q}}_{\alpha}
=\hat{\mathcal{Q}}_{i_{1},i_{2},\ldots,i_{k}}
=\mathcal{\hat{P}}_{i_{k}}^{(t_{k})}\odot\ldots\odot\mathcal{\hat{P}}_{i_{2}}^{(t_{2})}\odot\mathcal{\hat{P}}_{i_{1}}^{(t_{1})}
\end{equation}
 in \emph{history Hilbert space} $\breve{\mathcal{H}}=\mathcal{H}_{k}\odot\ldots\odot\mathcal{H}_{2}\odot\mathcal{H}_{1}$,
where $\mathcal{H}_{k}$ is a copy of the standard Hilbert space of
the physical system at time $t_{k}$, and the single index $\alpha\in\{1,2,\ldots,n^k\}$ is in one-to-one correspondence with each instantiation of the multi-index $\{i_1,i_2,\ldots,i_k\}$ using the explicit mapping
\begin{equation}
\alpha=1+(i_{1}-1)n^{0}+(i_{2}-1)n^{1}+\ldots+(i_{k}-1)n^{k-1}.
\end{equation}
The introduction of the single index $\alpha$ will be very useful in subsequent derivations because we can efficiently perform Feynman summation using a single sum rather than multiple $k$~sums. The result after summing over all $n^{k}$ orthogonal
quantum histories that span the history Hilbert space $\breve{\mathcal{H}}$
is the identity history \cite{Georgiev2018}

\begin{align}
\sum_{\alpha=1}^{n^{k}}\hat{\mathcal{Q}}_{\alpha}
&=\sum_{i_{1}=1}^{n}\sum_{i_{2}=1}^{n}\cdots\sum_{i_{k}=1}^{n}\hat{\mathcal{Q}}_{i_{1},i_{2},\ldots,i_{k}} \nonumber \\
&=\hat{I}_{k}\odot\ldots\odot\hat{I}_{2}\odot\hat{I}_{1}
.
\end{align}
\end{defn}

\begin{defn}
(Quantum histories with preselected and postselected quantum states)
The introduction of a complete set of quantum histories from a preselected
initial state $|\psi_{i}\rangle$ at~$t_{0}$ to a postselected final
state $|\psi_{f}\rangle$ at~$t_{k+1}$ is performed with the insertion
of initial projector $\mathcal{\hat{P}}_{0}=|\psi_{i}\rangle\langle\psi_{i}|$
and final projector $\mathcal{\hat{P}}_{k+1}=|\psi_{f}\rangle\langle\psi_{f}|$

\begin{equation}
\hat{\mathcal{Q}}_{\alpha}=\mathcal{\hat{P}}_{k+1}^{(t_{k}+1)}\odot\mathcal{\hat{P}}_{i_{k}}^{(t_{k})}\odot\ldots\odot\mathcal{\hat{P}}_{i_{2}}^{(t_{2})}\odot\mathcal{\hat{P}}_{i_{1}}^{(t_{1})}\odot\mathcal{\hat{P}}_{0}^{(t_{0})}
\end{equation}
such that the sum over all histories gives
\begin{equation}
\sum_{\alpha=1}^{n^{k}}\hat{\mathcal{Q}}_{\alpha}=\mathcal{\hat{P}}_{k+1}\odot\hat{I}_{k}\odot\ldots\odot\hat{I}_{2}\odot\hat{I}_{1}\odot\mathcal{\hat{P}}_{0}
,
\end{equation}
where we have dropped redundant time superscripts.
It is noteworthy that the superscript-subscript notation stacks information about multi-index instantiation vertically rather than horizontally in order to keep short the expressions for quantum histories, e.g.,
$\mathcal{\hat{P}}_{3}^{(t_{2})}\equiv\mathcal{\hat{P}}_{i_2=3}^{(t_{2})}\equiv\mathcal{\hat{P}}_{i_{2}=3}$.
When the multi-index is not instantiated, however, we can use the one-to-one correspondence between the indices and the time points, $i_{1}\leftrightarrow t_{1},i_{2}\leftrightarrow t_{2},\ldots,i_{k}\leftrightarrow t_{k}$, in order to drop the superscript, e.g. $\mathcal{\hat{P}}_{i_{2}}^{(t_{2})}\equiv\mathcal{\hat{P}}_{i_{2}}$.

\end{defn}

\begin{defn}
(Quantum history chain operators)
To each quantum history
$\hat{\mathcal{Q}}_{\alpha}=\mathcal{\hat{P}}_{k+1}\odot\mathcal{\hat{P}}_{i_{k}}\odot\ldots\odot\mathcal{\hat{P}}_{i_{2}}\odot\mathcal{\hat{P}}_{i_{1}}\odot\mathcal{\hat{P}}_{0}$
in history Hilbert space~$\breve{\mathcal{H}}$, there is a corresponding
chain operator $\hat{K}_{\alpha}$ in standard Hilbert space~$\mathcal{H}$
obtained from replacement of the time tensor symbols~$\odot$ with
the corresponding time evolution operators~$\hat{\mathcal{T}}$ as follows
\begin{equation}
\hat{K}_{\alpha}=\mathcal{\hat{P}}_{k+1}\hat{\mathcal{T}}(t_{k+1},t_{k})\ldots\mathcal{\hat{P}}_{i_{2}}\hat{\mathcal{T}}(t_{2},t_{1})\mathcal{\hat{P}}_{i_{1}}\hat{\mathcal{T}}(t_{1},t_{0})\mathcal{\hat{P}}_{0}
.
\end{equation}

Because the time evolution operators $\hat{\mathcal{T}}$ depend on
the concrete quantum Hamiltonian $\hat{H}$ of the dynamic quantum system, it
follows that the chain operators cannot be in general reused for different
experimental setups described with different quantum Hamiltonians.
In contrast, the same sets of quantum history projection operators
$\{\hat{\mathcal{Q}}_{\alpha}\}_{\alpha=1}^{n^{k}}$
in history Hilbert space~$\breve{\mathcal{H}}$ can be applied to different experimental setups
with different quantum Hamiltonians as long as the $n$-dimensional
Hilbert space of the system at each time point is identical.
\end{defn}

\begin{defn}
(Quantum history propagators) The Feynman propagator
$\psi_{\alpha}$
describing the flow of quantum probability amplitudes along a particular quantum history
$\hat{\mathcal{Q}_{\alpha}}$
in a physical system, which is preselected
in an initial quantum state $|\psi_{i}\rangle$ at $t_{0}$ and postselected
in a final quantum state $|\psi_{f}\rangle$ at $t_{k+1}$, is given
by the inner product involving the history chain operator \cite{Georgiev2018}

\begin{equation}
\psi_{\alpha}=\langle\psi_{f}|\hat{K}_{\alpha}|\psi_{i}\rangle
.
\end{equation}

\end{defn}

\begin{defn}
(Quantum histories of bipartite systems)
For bipartite quantum systems,
we construct a complete set of separable quantum histories in composite history Hilbert space
$\breve{\mathcal{H}}_{A\otimes B}$ using standard tensor products
$\otimes$ of local projection operators for each of the two systems $A$ and $B$ at each time point
\begin{equation}
\hat{\mathcal{Q}}=\left(\mathcal{\hat{P}}_{A_{k}}\otimes\mathcal{\hat{P}}_{B_{k}}\right)\odot\ldots\odot\left(\mathcal{\hat{P}}_{A_{2}}\otimes\mathcal{\hat{P}}_{B_{2}}\right)\odot\left(\mathcal{\hat{P}}_{A_{1}}\otimes\mathcal{\hat{P}}_{B_{1}}\right)
.
\end{equation}
Rearranging the tensor products allows for equivalent formulation in terms of the quantum histories for each component subsystem. Thus, alternatively we can also write
\begin{eqnarray}
\hat{\mathcal{Q}} & =& \left(\mathcal{\hat{P}}_{A_{k}}\odot\mathcal{\hat{P}}_{A_{2}}\odot\ldots\odot\mathcal{\hat{P}}_{A_{1}}\right)\nonumber \\
&& ~ \otimes\left(\mathcal{\hat{P}}_{B_{k}}\odot\mathcal{\hat{P}}_{B_{2}}\odot\ldots\odot\mathcal{\hat{P}}_{B_{1}}\right)\nonumber \\
 & = &\hat{\mathcal{Q}}_{A}\otimes\hat{\mathcal{Q}}_{B}
.
\end{eqnarray}
\end{defn}

\begin{defn}
(Schmidt decomposition based on bipartite Feynman propagators) The history
Hilbert space $\breve{\mathcal{H}}_{A}$ for the system $A$ and the
history Hilbert space $\breve{\mathcal{H}}_{B}$ for the system $B$
will each have $n^{k}$ orthogonal local quantum histories, resulting in
$n^{2k}$ orthogonal separable quantum histories for the composite history Hilbert
space $\breve{\mathcal{H}}_{A\otimes B}$. The rearrangement of the
tensor product of the separable bipartite quantum history into $\hat{\mathcal{Q}}_{A}\otimes\hat{\mathcal{Q}}_{B}$
allows for the introduction of a single index $\alpha\in\left\{ 1,\ldots n^{k}\right\} $
instead of the multiple index $\{A_{1},A_{2},\ldots,A_{k}\}$ and
single index $\beta\in\left\{ 1,\ldots n^{k}\right\} $ instead of
the multiple index $\{B_{1},B_{2},\ldots,B_{k}\}$. Consequently,
there will be $n^{2k}$ orthogonal composite quantum histories $\hat{\mathcal{Q}}_{\alpha,\beta}=\hat{\mathcal{Q}}_{\alpha}\otimes\hat{\mathcal{Q}}_{\beta}$
for which can be computed the bipartite Feynman propagators $\psi_{\alpha,\beta}=\langle\psi_{f}|\hat{K}_{\alpha,\beta}|\psi_{i}\rangle$.
This allows the construction of propagator complex coefficient matrix
\begin{equation}
\hat{C}=\left(\psi_{\alpha,\beta}\right)=\left(\begin{array}{cccc}
\psi_{1,1} & \psi_{1,2} & \ldots & \psi_{1,n^{k}}\\
\psi_{2,1} & \psi_{2,2} & \ldots & \psi_{2,n^{k}}\\
\vdots & \vdots & \ddots & \vdots\\
\psi_{n^{k},1} & \psi_{n^{k},2} & \ldots & \psi_{n^{k},n^{k}}
\end{array}\right)
\end{equation}
which can undergo singular value decomposition, $\hat{C}=\hat{U}\hat{\Lambda}\hat{V}^{\dagger}$
such that $\hat{U}$ and $\hat{V}^{\dagger}$ are unitary matrices,
and $\hat{\Lambda}$ is a diagonal matrix with non-negative singular
values sorted in descending order $\lambda_{1}\geq\lambda_{2}\geq\ldots\lambda_{s}\geq0$.

If the final postselected state occurs with unit probability, the singular values are referred to as Schmidt coefficients \cite{Georgiev2022}.
The squared Schmidt coefficients sum to unity, $\sum_i \lambda_i^2 =1$, and can be obtained as eigenvalues
of the following positive semidefinite Hermitian matrices
\begin{eqnarray}
\hat{C}\hat{C}^{\dagger} &=& \hat{U}\hat{\Lambda}\hat{V}^{\dagger}\hat{V}\hat{\Lambda}\hat{U}^{\dagger}=\hat{U}\hat{\Lambda}^{2}\hat{U}^{\dagger} ,\\
\hat{C}^{\dagger}\hat{C} &=& \hat{V}\hat{\Lambda}\hat{U}^{\dagger}\hat{U}\hat{\Lambda}\hat{V}^{\dagger}=\hat{V}\hat{\Lambda}^{2}\hat{V}^{\dagger} .
\end{eqnarray}

For the analysis of postselected states that occur with less than unit probability, the propagator complex coefficient matrix can be normalized using the square root of $\textrm{Tr}(\hat{C}\hat{C}^{\dagger})$, or more conveniently one can determine the squared Schmidt coefficients from
\begin{equation}
\tilde{C}=\frac{\hat{C}\hat{C}^{\dagger}}{\textrm{Tr}\left(\hat{C}\hat{C}^{\dagger}\right)}
.
\end{equation}
Due to the construction of the composite history Hilbert space $\breve{\mathcal{H}}_{A\otimes B}$ using mutually
orthogonal quantum history projection operators
\begin{equation}
\textrm{Tr}\left[\left(\hat{\mathcal{Q}}_{\alpha}\otimes\hat{\mathcal{Q}}_{\beta}\right)\cdot\left(\hat{\mathcal{Q}}_{\alpha^{\prime}}\otimes\hat{\mathcal{Q}}_{\beta^{\prime}}\right)\right]=\delta_{\alpha,\alpha^{\prime}}\delta_{\beta,\beta^{\prime}}
,
\end{equation}
the Feynman propagators $\psi_{\alpha,\beta}$ for each quantum history
give directly the quantum probabilities $\left|\psi_{\alpha,\beta}\right|^{2}$
for observing the corresponding quantum histories with sequential
strong measurements. This allows for simple expression of
the composite quantum history state vector $|\Psi)$ in $\breve{\mathcal{H}}_{A\otimes B}$ using the Feynman propagators as
\begin{equation}
|\Psi)=\sum_{\alpha}\sum_{\beta}\psi_{\alpha,\beta}|r_{\alpha})\otimes|r_{\beta}^{\prime})
\label{eq:canonical}
\end{equation}
where $|r_{\alpha})$ and $|r_{\beta}^{\prime})$ are the rays corresponding
to the local quantum history projectors $\hat{\mathcal{Q}}_{\alpha}=|r_{\alpha})(r_{\alpha}|$
and $\hat{\mathcal{Q}}_{\beta}=|r_{\beta}^{\prime})(r_{\beta}^{\prime}|$,
respectively for subsystems $A$ and $B$.

Introducing the matrix reshaping operation (row by row) into a column vector \cite{Miszczak2011}
\begin{equation}
\textrm{res}\left(\hat{X}\right)=\left(x_{1,1},x_{1,2},\ldots,x_{1,n},
\ldots,x_{n,1},x_{n,2},\ldots,x_{n,n}\right)^{T}
,
\end{equation}
it follows that the composite quantum history state vector $|\Psi)$
is obtained from the reshaped Feynman propagator complex coefficient
matrix $|\Psi)=\textrm{res}~\!(\hat{C})$.
Using the following matrix reshaping property \cite{Miszczak2011,Roth1934}
\begin{equation}
\textrm{res}\left(\hat{X}\hat{Y}\hat{Z}\right)=\left(\hat{X}\otimes\hat{Z}^{T}\right)\textrm{res}\left(\hat{Y}\right)
\end{equation}
applied to \eqref{eq:canonical}, we obtain
\begin{equation}
|\Psi) =\textrm{res}\left(\hat{U}\hat{\Lambda}\hat{V}^{\dagger}\right)=\left(\hat{U}\otimes\hat{V}^{*}\right)\textrm{res}\left(\hat{\Lambda}\right)
.
\end{equation}
Since $\hat{U}$ and $\hat{V}^{*}$ are local unitary operators, the action of their tensor product preserves the basis separability during the
change of the original basis $|r_{\alpha})$, $|r_{\beta}^{\prime})$ into Schmidt basis
$|\tilde{\alpha}_{s})=\hat{U}|r_{s})$, $|\tilde{\beta}_{s})=\hat{V}^{*}|r_{s}^{\prime})$ as follows
\begin{align}
|\Psi) & =\left(\hat{U}\otimes\hat{V}^{*}\right)\sum_{\alpha}\sum_{\beta}\lambda_{\alpha,\beta}\delta_{\alpha,\beta}|r_{\alpha})\otimes|r_{\beta}^{\prime})\nonumber \\
& =\sum_{s}\lambda_{s}\left[\hat{U}|r_{s})\right]\otimes\left[\hat{V}^{*}|r_{s}^{\prime})\right] \nonumber \\
& =\sum_{s}\lambda_{s}|\tilde{\alpha}_{s})\otimes|\tilde{\beta}_{s})
.
\label{eq:Schmidt}
\end{align}
\end{defn}

\begin{defn}
(Schmidt rank) The~Schmidt rank of a particular Schmidt decomposition is the number of nonzero Schmidt coefficients \cite{Miszczak2011}.
\end{defn}

\begin{thm}
(Entanglement of quantum histories)
The~Schmidt rank provides a discrete (binary) criterion for entanglement, namely, the quantum histories are entangled if and only if the Schmidt rank of the propagator complex coefficient matrix $\hat{C}$ is larger than 1.
\end{thm}
\begin{proof}
If the Schmidt rank of $|\Psi)$ is 1, then its Schmidt decomposition \eqref{eq:Schmidt} contains only a single non-zero Schmidt
coefficient $\lambda_{1}=1$, $\lambda_{2}=\lambda_{3}=\ldots=\lambda_{s}=0$.
Thus, the composite quantum history state vector $|\Psi)$ is separable,
$|\Psi)=1\,|\tilde{\alpha}_{1})\otimes|\tilde{\beta}_{1})$, hence
not entangled.

Conversely, if $|\Psi)$ is separable, then its Schmidt rank is~1.
Indeed, from the definition of separability, it follows that $|\Psi)$
can be written as $|\tilde{\alpha})\otimes|\tilde{\beta})$ for some
separable basis. In this basis, the matrix $\hat{C}$ is already diagonal and
one can observe that there is only a single non-zero Schmidt coefficient
$\lambda_{1}=1$. Because the Schmidt coefficients are unique, the change of basis for expressing the state $|\Psi)$
will not increase the number of non-zero Schmidt coefficients. Therefore, if the Schmidt rank is greater than 1, then $|\Psi)$
is necessarily entangled.
\end{proof}

Weak measurement is a fruitful technique, which can extract information about the properties of a quantum state without collapsing the state into eigenvectors of the measured operator \cite{Aharonov1988,Aharonov2014,Tamir2013}. This is done by creating a weak coupling between the dynamically evolving quantum system and a measurement device whose pointer is usually prepared with an initial Gaussian wave function centered at zero
\begin{equation}
\phi(x)=\left(2\pi\sigma^{2}\right)^{-\frac{1}{4}} e^{-\frac{x^{2}}{4\sigma^{2}}}
\end{equation}
If the weak interaction Hamiltonian between the measured system and the measuring device is
\begin{equation}
\hat{H}_{\textrm{int}}=g\delta(t-t_1)\,\hat{\mathcal{P}}_{i_1}\otimes\hat{p}
\end{equation}
where $g$ is the weak coupling parameter, $\hat{\mathcal{P}}_{i_1}$ is a projection operator (observable) for the measured system and
$\hat{p}=\hbar\hat{k}$ is
the pointer momentum conjugate to the position $\hat{x}$, it is possible to extract the weak value~\cite{Georgiev2018}
\begin{equation}
\left(\mathcal{\hat{P}}_{i_{1}}\right)_{w}=\frac{\langle\psi_{f}|\hat{\mathcal{T}}(t_{2},t_{1})\mathcal{\hat{P}}_{i_{1}}\hat{\mathcal{T}}(t_{1},t_{0})|\psi_{i}\rangle}{\langle\psi_{f}|\hat{\mathcal{T}}(t_{2},t_{1})\hat{\mathcal{T}}(t_{1},t_{0})|\psi_{i}\rangle}
\end{equation}
using strong projective measurement of the position $\hat{x}$ or the wave number $\hat{k}$ of the meter pointer after the quantum system of interest is postselected in the desired final state $\vert \psi_f\rangle$.
The expectation (average) values of the latter quantities approximate to first order in $g$ the real and imaginary parts of the weak value,
$\textrm{Re}[(\mathcal{\hat{P}}_{i_{1}})_{w}] \approx  g^{-1} \langle \hat{x}\rangle $
and
$\textrm{Im}[(\mathcal{\hat{P}}_{i_{1}})_{w}] \approx  2\sigma^2 g^{-1} \langle \hat{k}\rangle $
\cite{Jozsa2007,Dressel2014,Mitchison2007}.

Sequential weak values generalize the notion of weak value to a sequence of multi-time observables. For quantum histories, we consider the special case when the observables are multi-time projection operators
\begin{equation}
\left(Q_{\alpha^{\prime}}\right)_{w}=\left(\mathcal{\hat{P}}_{i_{k}},\ldots,\mathcal{\hat{P}}_{i_{2}},\mathcal{\hat{P}}_{i_{1}}\right)_{w}=\frac{\psi_{\alpha^{\prime}}}{\sum_{\alpha}\psi_{\alpha}}
.
\end{equation}
If each observable is weakly coupled to its own measuring device with a Gaussian pointer, then sequential weak values can be extracted by measuring the expectation values of various product combinations of pointer positions and momenta as described elsewhere \cite{Mitchison2007}.
For our purposes, it is sufficient to note that sequential weak values are given by Feynman propagator ratios and can be experimentally determined using weak measurements \cite{Piacentini2016,Kim2018}.

\begin{defn}
(Sequential weak value of bipartite quantum histories)
The sequential weak value $(Q_{\alpha^{\prime},\beta^{\prime}})_w$ of the projection operators generating the particular bipartite quantum history $\hat{\mathcal{Q}}_{\alpha^{\prime},\beta^{\prime}}$ is given by the ratio of the Feynman propagator $\psi_{\alpha^{\prime},\beta^{\prime}}$ and the sum over all quantum histories that start from the same initial state and end with the same final state \cite{Georgiev2018}
\begin{equation}
\left(Q_{\alpha^{\prime},\beta^{\prime}}\right)_{w}=\frac{\psi_{\alpha^{\prime},\beta^{\prime}}}{\sum_{\alpha,\beta}\psi_{\alpha,\beta}}
\end{equation}
If the histories contain only a single intermediate time point, the sequential weak values reduce to weak values.
\end{defn}

\begin{thm}
(Schmidt decomposition based on sequential weak values)
The Schmidt coefficients for quantification of entanglement of quantum histories can be determined from the complex sequential weak value matrix after normalization
\begin{equation}
\hat{M}=\left(Q_{\alpha,\beta}\right)_{w}=\left(\begin{array}{cccc}
\left(Q_{1,1}\right)_{w} & \left(Q_{1,2}\right)_{w} & \ldots & \left(Q_{1,n^{k}}\right)_{w}\\
\left(Q_{2,1}\right)_{w} & \left(Q_{2,2}\right)_{w} & \ldots & \left(Q_{2,n^{k}}\right)_{w}\\
\vdots & \vdots & \ddots & \vdots\\
\left(Q_{n^{k},1}\right)_{w} & \left(Q_{n^{k},2}\right)_{w} & \ldots & \left(Q_{n^{k},n^{k}}\right)_{w}
\end{array}\right)
\end{equation}
\end{thm}
\begin{proof}
The Feynman propagator complex coefficient matrix $\hat{C}$ differs from $\hat{M}$ by a pure phase and a multiplicative constant $\left| \sum_{\alpha^{\prime},\beta^{\prime}}\psi_{\alpha^{\prime},\beta^{\prime}}\right|$. The pure phase is eliminated when the Hermitian matrix $\hat{M}\hat{M}^\dagger$ is generated, whereas the multiplicative constant is taken care of during the normalization
\begin{equation}
\tilde{M}=\frac{\hat{M}\hat{M}^{\dagger}}{\textrm{Tr}\left(\hat{M}\hat{M}^{\dagger}\right)}
\end{equation}
Therefore, the eigenvalues of $\tilde{M}$ are exactly the squared Schmidt coefficients
$\lambda_{1}^{2}\geq\lambda_{2}^{2}\geq\ldots\lambda_{s}^{2}\geq0$.
\end{proof}

\begin{defn}
(Concurrence) Concurrence provides a continuous quantitative measure of entanglement
based on the sum of the fourth powers of the Schmidt coefficients
\cite{Rungta2001,Georgiev2022}
\begin{equation}
\label{eq:concurrence}
\mathcal{C}=\sqrt{2\left(1-\sum_{s}\lambda_{s}^{4}\right)}=\sqrt{2\left[1-\textrm{Tr}\left(\tilde{C}^{2}\right)\right]}
\end{equation}
where $n=\min\left[\textrm{dim}\left(\breve{\mathcal{H}}_{A}\right),\textrm{dim}\left(\breve{\mathcal{H}}_{B}\right)\right]$
is the dimension of the smaller history Hilbert space of the systems~$A$~or~$B$.
\end{defn}

\begin{defn}
(Entanglement entropy) The entanglement entropy is defined as the
Shannon entropy \cite{Shannon1948a} of the squared Schmidt
coefficients \cite{Vedral1997}
\begin{equation}
\mathcal{S}=-\sum_{s}\lambda_{s}^{2}\ln\lambda_{s}^{2}
.
\end{equation}
\end{defn}

\begin{defn}
(Entanglement robustness) The entanglement robustness is given by
the difference between the squared sum of the Schmidt coefficients
and the sum of the squared Schmidt coefficients \cite{Vidal1999}
\begin{equation}
\mathcal{R}=\left(\sum_{s}\lambda_{s}\right)^{2}-\sum_{s}\lambda_{s}^{2}=\left(\sum_{s}\lambda_{s}\right)^{2}-1
.
\end{equation}
\end{defn}

Next, we will elaborate on the construction
of the Feynman propagator complex coefficient matrix $\hat{C}$ and
will compute concurrence for different postselections in Hardy's interferometer. This choice of entanglement measure is motivated
by the manifestation of concurrence in the form of two-particle visibility
\cite{Georgiev2021} and its intimate involvement in quantum complementarity
relations \cite{Jakob2010}. It is worth noting, however, that the
Schmidt coefficients $\lambda_{s}$, which are obtained through singular value
decomposition of the Feynman propagator complex coefficient matrix~$\hat{C}$,
also contain all the necessary information for the calculation
of a large number of other quantum information-theoretic entanglement measures
\cite{Georgiev2022,Gudder2020b,Androulakis2020}.

\section{\label{sec:3}Quantum dynamics of the composite state in Hardy's interferometer}

Lucien Hardy proposed in 1992 an interferometric experiment with an
electron and a positron in order to test the predictions of quantum
mechanics versus the predictions of local hidden variable theories
\cite{Hardy1992}.

\begin{figure*}[t!]
\begin{centering}
\includegraphics[width=110mm]{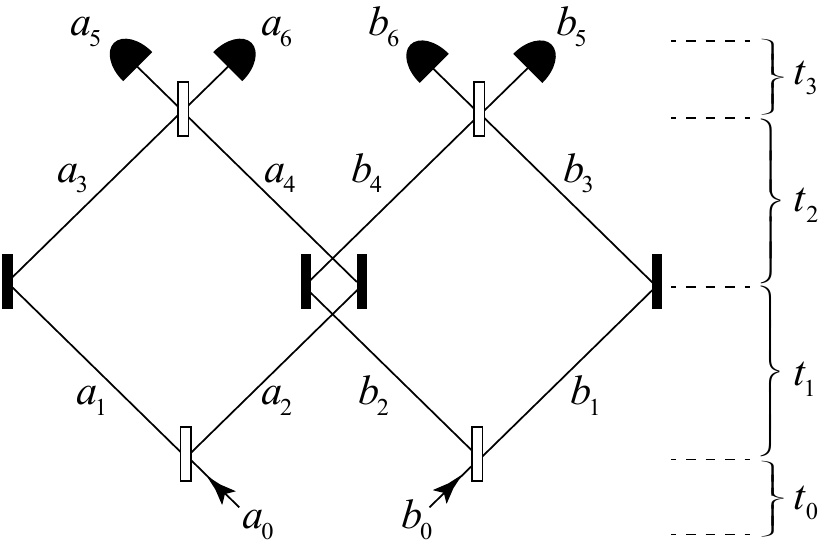}
\par\end{centering}

\caption{\label{fig:1}Hardy's interferometer is constructed from two partially overlapping Mach--Zehnder interferometers, one of which contains an electron (a) and the other contains a positron (b) \cite{Hardy1992}. Possible particle trajectories are indicated with solid lines.
Beamsplitters are indicated with white bars, whereas mirrors are indicated with black bars.}
\end{figure*}

Suppose that at time $t_{0}$ we have an electron at position $a_{0}$
and a positron at position $b_{0}$, both of which are aimed at the
corresponding beamsplitters in the overlapping Mach--Zehnder interferometers
(Fig.~\ref{fig:1}). Employing the formalism of second quantization \cite{Fock1932},
we will introduce electron creation $\hat{a}^{\dagger}$ and annihilation $\hat{a}$ operators or
positron creation $\hat{b}^{\dagger}$ and annihilation $\hat{b}$ operators that act on the vacuum
ket state $|0\rangle$ or the vacuum bra state $\langle 0|$ to produce corresponding Fock states.
Thus, the initial composite quantum state is
\begin{equation}
|\Psi_{0}\rangle=\hat{a}_{0}^{\dagger}\hat{b}_{0}^{\dagger}|0\rangle=|a_{0}\rangle|b_{0}\rangle
.
\end{equation}
At time $t_{1}$, the action of the beamsplitters is given by
\begin{eqnarray}
\hat{a}_{0}^{\dagger}|0\rangle & \rightarrow & \frac{1}{\sqrt{2}}\left(\hat{a}_{1}^{\dagger}+\imath\hat{a}_{2}^{\dagger}\right)|0\rangle ,\\
\hat{b}_{0}^{\dagger}|0\rangle & \rightarrow & \frac{1}{\sqrt{2}}\left(\hat{b}_{1}^{\dagger}+\imath\hat{b}_{2}^{\dagger}\right)|0\rangle .
\end{eqnarray}
The time evolution operator is
\begin{equation}
\hat{\mathcal{T}}(t_{1},t_{0})=\frac{1}{2}\left(\hat{a}_{1}^{\dagger}+\imath\hat{a}_{2}^{\dagger}\right)\left(\hat{b}_{1}^{\dagger}+\imath\hat{b}_{2}^{\dagger}\right)|0\rangle\langle0|\hat{a}_{0}\hat{b}_{0}
.
\end{equation}
Because the collision of electron and positron leads to their annihilation
and production of photons \cite{Colgate1953} that will not be registered by the final particle
detectors, we will take $\hat{a}_{2}^{\dagger}\hat{b}_{2}^{\dagger}|0\rangle=\hat{a}_{4}^{\dagger}\hat{b}_{4}^{\dagger}|0\rangle=0$
with resulting $\frac{1}{4}$ drop of total probability for particle detection.

Straightforward algebraic calculation gives the resulting composite quantum state at $t_{1}$ as
\begin{equation}
|\Psi_{1}\rangle=\frac{1}{2}\left(\hat{a}_{1}^{\dagger}\hat{b}_{1}^{\dagger}+\imath\hat{a}_{1}^{\dagger}\hat{b}_{2}^{\dagger}+\imath\hat{a}_{2}^{\dagger}\hat{b}_{1}^{\dagger}\right)|0\rangle
,
\end{equation}
which evolves after the mirror reflections
\begin{eqnarray}
\hat{a}_{1}^{\dagger}|0\rangle & \rightarrow & \imath\hat{a}_{3}^{\dagger}|0\rangle ,\\
\hat{a}_{2}^{\dagger}|0\rangle & \rightarrow & \imath\hat{a}_{4}^{\dagger}|0\rangle ,\\
\hat{b}_{1}^{\dagger}|0\rangle & \rightarrow & \imath\hat{b}_{3}^{\dagger}|0\rangle ,\\
\hat{b}_{2}^{\dagger}|0\rangle & \rightarrow & \imath\hat{b}_{4}^{\dagger}|0\rangle ,
\end{eqnarray}
with time evolution operator 
\begin{eqnarray}
\hat{\mathcal{T}}(t_{2},t_{1}) &=&-\hat{a}_{3}^{\dagger}\hat{b}_{3}^{\dagger}|0\rangle\langle0|\hat{a}_{1}\hat{b}_{1}-\hat{a}_{3}^{\dagger}\hat{b}_{4}^{\dagger}|0\rangle\langle0|\hat{a}_{1}\hat{b}_{2}\nonumber\\
&&  -\hat{a}_{4}^{\dagger}\hat{b}_{3}^{\dagger}|0\rangle\langle0|\hat{a}_{2}\hat{b}_{1}-\hat{a}_{4}^{\dagger}\hat{b}_{4}^{\dagger}|0\rangle\langle0|\hat{a}_{2}\hat{b}_{2}
\end{eqnarray}
into
\begin{equation}
|\Psi_{2}\rangle=-\frac{1}{2}\left(\hat{a}_{3}^{\dagger}\hat{b}_{3}^{\dagger}+\imath\hat{a}_{3}^{\dagger}\hat{b}_{4}^{\dagger}+\imath\hat{a}_{4}^{\dagger}\hat{b}_{3}^{\dagger}\right)|0\rangle
.
\label{eq:2}
\end{equation}
Finally, at $t_{3}$ after the action of the final beamsplitters
\begin{eqnarray}
\hat{a}_{3}^{\dagger}|0\rangle & \rightarrow & \frac{1}{\sqrt{2}}\left(\imath\hat{a}_{5}^{\dagger}+\hat{a}_{6}^{\dagger}\right)|0\rangle ,\\
\hat{a}_{4}^{\dagger}|0\rangle & \rightarrow & \frac{1}{\sqrt{2}}\left(\hat{a}_{5}^{\dagger}+\imath\hat{a}_{6}^{\dagger}\right)|0\rangle ,\\
\hat{b}_{3}^{\dagger}|0\rangle & \rightarrow & \frac{1}{\sqrt{2}}\left(\imath\hat{b}_{5}^{\dagger}+\hat{b}_{6}^{\dagger}\right)|0\rangle ,\\
\hat{b}_{4}^{\dagger}|0\rangle & \rightarrow & \frac{1}{\sqrt{2}}\left(\hat{b}_{5}^{\dagger}+\imath\hat{b}_{6}^{\dagger}\right)|0\rangle ,
\end{eqnarray}
with time evolution operator
\begin{widetext}
\begin{eqnarray}
\hat{\mathcal{T}}(t_{3},t_{2}) &=& \frac{1}{2}\left(\imath\hat{a}_{5}^{\dagger}+\hat{a}_{6}^{\dagger}\right)\left(\imath\hat{b}_{5}^{\dagger}+\hat{b}_{6}^{\dagger}\right)|0\rangle\langle0|\hat{a}_{3}\hat{b}_{3} 
 +\frac{1}{2}\left(\imath\hat{a}_{5}^{\dagger}+\hat{a}_{6}^{\dagger}\right)\left(\hat{b}_{5}^{\dagger}+\imath\hat{b}_{6}^{\dagger}\right)|0\rangle\langle0|\hat{a}_{3}\hat{b}_{4} \nonumber \\
&& +\frac{1}{2}\left(\hat{a}_{5}^{\dagger}+\imath\hat{a}_{6}^{\dagger}\right)\left(\imath\hat{b}_{5}^{\dagger}+\hat{b}_{6}^{\dagger}\right)|0\rangle\langle0|\hat{a}_{4}\hat{b}_{3} 
 +\frac{1}{2}\left(\hat{a}_{5}^{\dagger}+\imath\hat{a}_{6}^{\dagger}\right)\left(\hat{b}_{5}^{\dagger}+\imath\hat{b}_{6}^{\dagger}\right)|0\rangle\langle0|\hat{a}_{4}\hat{b}_{4}
,
\end{eqnarray}
the particles arrive at the detectors in the composite state
\begin{align}
|\Psi_{3}\rangle & =-\frac{1}{4}\left[\left(\imath\hat{a}_{5}^{\dagger}+\hat{a}_{6}^{\dagger}\right)\left(\imath\hat{b}_{5}^{\dagger}+\hat{b}_{6}^{\dagger}\right)+\imath\left(\imath\hat{a}_{5}^{\dagger}+\hat{a}_{6}^{\dagger}\right)\left(\hat{b}_{5}^{\dagger}+\imath\hat{b}_{6}^{\dagger}\right)+\imath\left(\hat{a}_{5}^{\dagger}+\imath\hat{a}_{6}^{\dagger}\right)\left(\imath\hat{b}_{5}^{\dagger}+\hat{b}_{6}^{\dagger}\right)\right]|0\rangle\nonumber \\
 & =\frac{1}{4}\left[\left(\hat{a}_{5}^{\dagger}\hat{b}_{5}^{\dagger}-\imath\hat{a}_{6}^{\dagger}\hat{b}_{5}^{\dagger}-\imath\hat{a}_{5}^{\dagger}\hat{b}_{6}^{\dagger}-\hat{a}_{6}^{\dagger}\hat{b}_{6}^{\dagger}\right)+\left(\hat{a}_{5}^{\dagger}\hat{b}_{5}^{\dagger}-\imath\hat{a}_{6}^{\dagger}\hat{b}_{5}^{\dagger}+\imath\hat{a}_{5}^{\dagger}\hat{b}_{6}^{\dagger}+\hat{a}_{6}^{\dagger}\hat{b}_{6}^{\dagger}\right)+\left(\hat{a}_{5}^{\dagger}\hat{b}_{5}^{\dagger}+\imath\hat{a}_{6}^{\dagger}\hat{b}_{5}^{\dagger}-\imath\hat{a}_{5}^{\dagger}\hat{b}_{6}^{\dagger}+\hat{a}_{6}^{\dagger}\hat{b}_{6}^{\dagger}\right)\right]|0\rangle\nonumber \\
 & =\frac{1}{4}\left(3\hat{a}_{5}^{\dagger}\hat{b}_{5}^{\dagger}-\imath\hat{a}_{6}^{\dagger}\hat{b}_{5}^{\dagger}-\imath\hat{a}_{5}^{\dagger}\hat{b}_{6}^{\dagger}+\hat{a}_{6}^{\dagger}\hat{b}_{6}^{\dagger}\right)|0\rangle\label{eq:3}
.
\end{align}
The quantum probabilities of simultaneous detection of an electron and a positron
at the four particles detectors are $P\left(a_{5},b_{5}\right)=\frac{9}{16}$,
$P\left(a_{5},b_{6}\right)=\frac{1}{16}$,
$P\left(a_{6},b_{5}\right)=\frac{1}{16}$
and $P\left(a_{6},b_{6}\right)=\frac{1}{16}$.
The sum of probabilities for detection of an electron and a positron
at the end of the experiment is $\frac{3}{4}$ because the
electron and positron annihilate inside the interferometer in
$\frac{1}{4}$ of all cases.

\section{\label{sec:4}Classical histories are unable to reproduce the quantum
outcomes}

The probability for the electron to arrive at $a_{5}$ or $a_{6}$,
when the positron is detected at $b_{3}$ or $b_{4}$, can be computed
using the Born rule applied on the final composite state vector
\begin{align}
|\Psi_{3}^{(a)}\rangle & =-\frac{1}{\sqrt{8}}\left[\left(\imath\hat{a}_{5}^{\dagger}+\hat{a}_{6}^{\dagger}\right)\hat{b}_{3}^{\dagger}+\imath\left(\imath\hat{a}_{5}^{\dagger}+\hat{a}_{6}^{\dagger}\right)\hat{b}_{4}^{\dagger}+\imath\left(\hat{a}_{5}^{\dagger}+\imath\hat{a}_{6}^{\dagger}\right)\hat{b}_{3}^{\dagger}\right]|0\rangle\nonumber \\
 & =-\frac{1}{\sqrt{8}}\left[\imath\hat{a}_{5}^{\dagger}\hat{b}_{3}^{\dagger}+\hat{a}_{6}^{\dagger}\hat{b}_{3}^{\dagger}-\hat{a}_{5}^{\dagger}\hat{b}_{4}^{\dagger}+\imath\hat{a}_{6}^{\dagger}\hat{b}_{4}^{\dagger}+\imath\hat{a}_{5}^{\dagger}\hat{b}_{3}^{\dagger}-\hat{a}_{6}^{\dagger}\hat{b}_{3}^{\dagger}\right]|0\rangle\nonumber \\
 & =-\frac{1}{\sqrt{8}}\left[2\imath\hat{a}_{5}^{\dagger}\hat{b}_{3}^{\dagger}-\hat{a}_{5}^{\dagger}\hat{b}_{4}^{\dagger}+\imath\hat{a}_{6}^{\dagger}\hat{b}_{4}^{\dagger}\right]|0\rangle\label{eq:a}
.
\end{align}
The probabilities of simultaneous detection of the electron at the
final detectors and the positron inside the interferometer arms is
$P\left(a_{5},b_{3}\right)=\frac{1}{2}$, $P\left(a_{5},b_{4}\right)=P\left(a_{6},b_{4}\right)=\frac{1}{8}$
and $P\left(a_{6},b_{3}\right)=0$. This implies that the conditional
probabilities \cite{Kolmogorov1956} for positron detection at $b_{3}$ or $b_{4}$ given
electron postselection at $a_{6}$ are
\begin{align}
P\left(b_{4}|a_{6}\right) & =\frac{P\left(a_{6},b_{4}\right)}{P\left(a_{6},b_{4}\right)+P\left(a_{6},b_{3}\right)}=1 ,\\
P\left(b_{3}|a_{6}\right) & =\frac{P\left(a_{6},b_{3}\right)}{P\left(a_{6},b_{4}\right)+P\left(a_{6},b_{3}\right)}=0 .
\end{align}
Similarly, the probability for the positron to arrive at $b_{5}$
or $b_{6}$, when the electron is detected at $a_{3}$ or $a_{4}$,
can be computed using the Born rule applied on the final composite state vector
\begin{align}
|\Psi_{3}^{(b)}\rangle & =-\frac{1}{\sqrt{8}}\left[\hat{a}_{3}^{\dagger}\left(\imath\hat{b}_{5}^{\dagger}+\hat{b}_{6}^{\dagger}\right)+\imath\hat{a}_{3}^{\dagger}\left(\hat{b}_{5}^{\dagger}+\imath\hat{b}_{6}^{\dagger}\right)+\imath\hat{a}_{4}^{\dagger}\left(\imath\hat{b}_{5}^{\dagger}+\hat{b}_{6}^{\dagger}\right)\right]|0\rangle\nonumber \\
 & =-\frac{1}{\sqrt{8}}\left[\imath\hat{a}_{3}^{\dagger}\hat{b}_{5}^{\dagger}+\hat{a}_{3}^{\dagger}\hat{b}_{6}^{\dagger}+\imath\hat{a}_{3}^{\dagger}\hat{b}_{5}^{\dagger}-\hat{a}_{3}^{\dagger}\hat{b}_{6}^{\dagger}-\hat{a}_{4}^{\dagger}\hat{b}_{5}^{\dagger}+\imath\hat{a}_{4}^{\dagger}\hat{b}_{6}^{\dagger}\right]|0\rangle\nonumber \\
 & =-\frac{1}{\sqrt{8}}\left[2\imath\hat{a}_{3}^{\dagger}\hat{b}_{5}^{\dagger}-\hat{a}_{4}^{\dagger}\hat{b}_{5}^{\dagger}+\imath\hat{a}_{4}^{\dagger}\hat{b}_{6}^{\dagger}\right]|0\rangle\label{eq:b}
.
\end{align}
\end{widetext}
The probabilities of simultaneous detection of the positron at the
final detectors and the electron inside the interferometer arms is
$P\left(a_{3},b_{5}\right)=\frac{1}{2}$, $P\left(a_{4},b_{5}\right)=P\left(a_{4},b_{6}\right)=\frac{1}{8}$
and $P\left(a_{3},b_{6}\right)=0$. This implies that the conditional
probabilities for electron detection at $a_{3}$ or $a_{4}$ given
positron postselection at $b_{6}$ are
\begin{align}
P\left(a_{4}|b_{6}\right) & =\frac{P\left(a_{4},b_{6}\right)}{P\left(a_{4},b_{6}\right)+P\left(a_{3},b_{6}\right)}=1 ,\\
P\left(a_{3}|b_{6}\right) & =\frac{P\left(a_{3},b_{6}\right)}{P\left(a_{4},b_{6}\right)+P\left(a_{3},b_{6}\right)}=0 .
\end{align}
Now, if one conjectures that the electron and the positron could have
traveled along a single classical path along one of two alternative
arms in the corresponding interferometers, a contradiction will occur
as follows: From the electron detection at $a_{6}$ one uses $P\left(b_{4}|a_{6}\right)=1$
to infer that the positron has certainly passed along path $b_{4}$.
Similarly, from the positron detection at $b_{6}$ one uses $P\left(a_{4}|b_{6}\right)=1$
to infer that the electron has certainly passed along path $a_{4}$.
However, the collision of the electron and positron at the crossing
of arms $a_{4}$ and $b_{4}$ leads to zero probability of particle
detection at $a_{6}$ and $b_{6}$ 
\begin{equation}
P\left(a_{6},b_{6}|a_{4},b_{4}\right)=0
.
\end{equation}
Arriving at a contradiction, proves that the conjecture of quantum
particles moving along single classical paths is false. Indeed, next we
will show that the predictions of quantum mechanics arise exactly
because the quantum histories can be entangled.

\section{\label{sec:5}Illustration of entangled quantum histories}

Quantum analysis of Hardy's interferometer from the initial (preselected) state
$|a_{0}\rangle|b_{0}\rangle$ to final (postselected) state $|a_{6}\rangle|b_{6}\rangle$
reveals a 4-dimensional composite quantum history Hilbert space $\breve{\mathcal{H}}_{A\otimes B}$,
where individual systems explore 2-dimensional quantum history Hilbert
spaces $\breve{\mathcal{H}}_{A}$ and $\breve{\mathcal{H}}_{B}$.
In order to ease the notation, we will compress the projectors
as follows
\begin{equation}
\hat{\mathcal{P}}\left(a_{i}b_{j}\right):=\hat{a}_{i}^{\dagger}\hat{b}_{j}^{\dagger}|0\rangle\langle0|\hat{a}_{i}\hat{b}_{j}
=|a_{i}\rangle|b_{j}\rangle\langle a_{i}|\langle b_{j}|
.
\end{equation}
Then, the four composite quantum histories are
\begin{align}
\hat{\mathcal{Q}}_{1,1} & =\hat{\mathcal{P}}\left(a_{6}b_{6}\right)\odot\hat{\mathcal{P}}\left(a_{3}b_{3}\right)\odot\hat{\mathcal{P}}\left(a_{1}b_{1}\right)\odot\hat{\mathcal{P}}\left(a_{0}b_{0}\right) ,\\
\hat{\mathcal{Q}}_{1,2} & =\hat{\mathcal{P}}\left(a_{6}b_{6}\right)\odot\hat{\mathcal{P}}\left(a_{3}b_{4}\right)\odot\hat{\mathcal{P}}\left(a_{1}b_{2}\right)\odot\hat{\mathcal{P}}\left(a_{0}b_{0}\right) ,\\
\hat{\mathcal{Q}}_{2,1} & =\hat{\mathcal{P}}\left(a_{6}b_{6}\right)\odot\hat{\mathcal{P}}\left(a_{4}b_{3}\right)\odot\hat{\mathcal{P}}\left(a_{2}b_{1}\right)\odot\hat{\mathcal{P}}\left(a_{0}b_{0}\right) ,\\
\hat{\mathcal{Q}}_{2,2} & =\hat{\mathcal{P}}\left(a_{6}b_{6}\right)\odot\hat{\mathcal{P}}\left(a_{4}b_{4}\right)\odot\hat{\mathcal{P}}\left(a_{2}b_{2}\right)\odot\hat{\mathcal{P}}\left(a_{0}b_{0}\right) .
\end{align}

\begin{figure*}[t!]
\begin{centering}
\includegraphics[width=150mm]{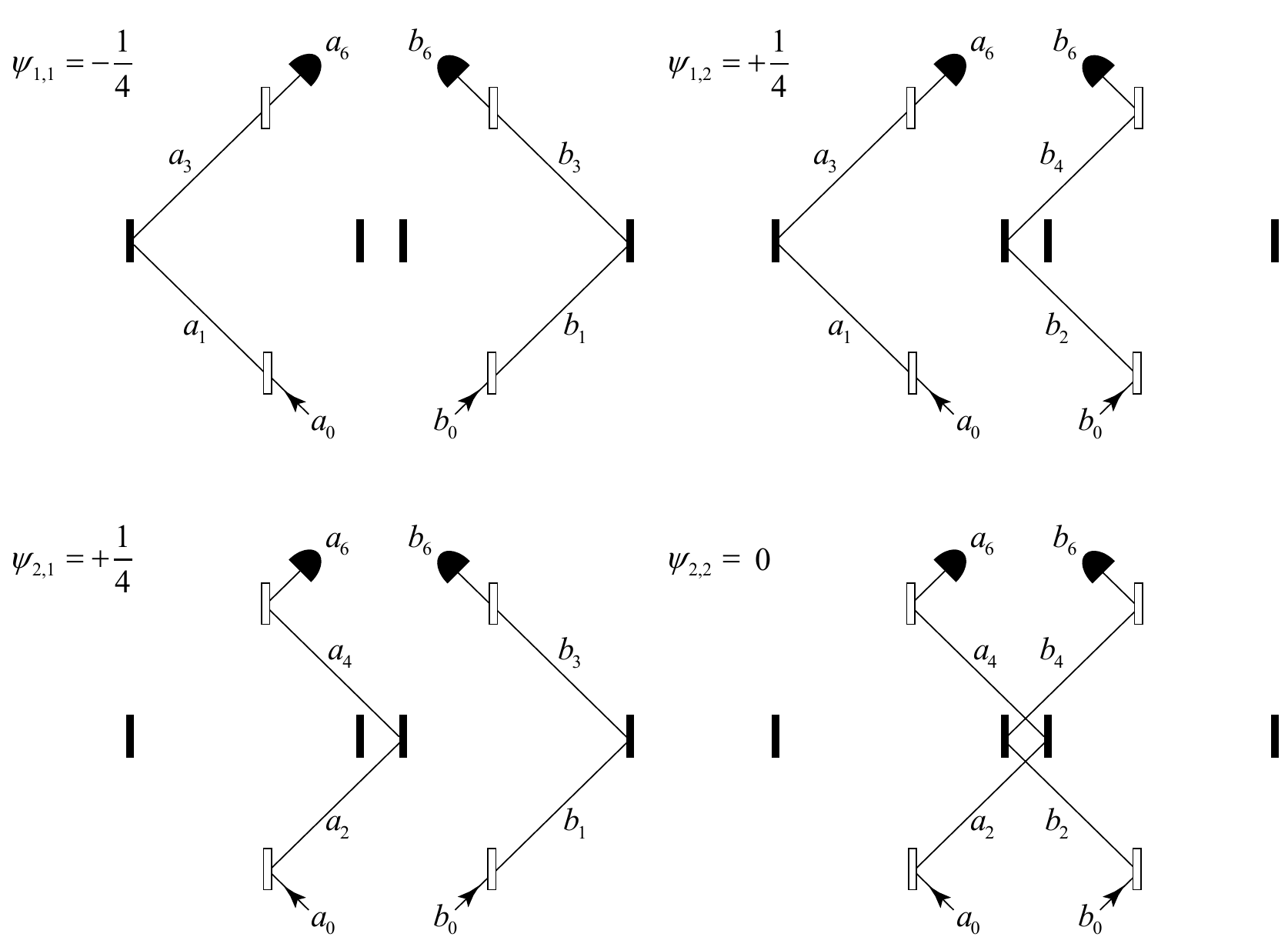}
\par\end{centering}

\caption{\label{fig:2}Illustration of the four possible composite quantum
histories, which lead to electron emerging at $a_{6}$ and positron
emerging at $b_{6}$, together with the corresponding Feynman propagators
$\psi_{\alpha,\beta}$. Due to electron-positron annihilation, one
of the histories has zero Feynman propagator, which leaves a set of
3 non-zero quantum histories in quantum entangled superposition.}
\end{figure*}

\noindent The Feynman propagators $\psi_{\alpha,\beta}=\langle a_{6}|\langle b_{6}|\hat{K}_{\alpha,\beta}|a_{0}\rangle|b_{0}\rangle$,
in which $\hat{K}_{\alpha,\beta}$ is the chain operator corresponding to the quantum
history $\hat{\mathcal{Q}}_{\alpha,\beta}$, can be arranged in the
propagator complex coefficient matrix as follows (Fig.~\ref{fig:2})
\begin{equation}
\hat{C}=\left(\begin{array}{cc}
\psi_{1,1} & \psi_{1,2}\\
\psi_{2,1} & \psi_{2,2}
\end{array}\right)=\frac{1}{4}\left(\begin{array}{cc}
-1 & 1\\
1 & 0
\end{array}\right)
.
\end{equation}
Because we are interested in postselected outcome, normalization can
be easily performed at the stage of computing the Hermitian matrix
$\hat{C}\hat{C}^{\dagger}$ where we use the fact that the squared
Schmidt coefficients sum to unity
\begin{equation}
\tilde{C}=\frac{\hat{C}\hat{C}^{\dagger}}{\textrm{Tr}\left(\hat{C}\hat{C}^{\dagger}\right)}=\frac{1}{3}\left(\begin{array}{cc}
2 & -1\\
-1 & 1
\end{array}\right)
.
\end{equation}
The two eigenvalues expressed as squared Schmidt coefficients are
\begin{eqnarray}
\lambda_{1}^{2} &=& \frac{1}{6}\left(3+\sqrt{5}\right), \label{eq:S1}\\
\lambda_{2}^{2} &=& \frac{1}{6}\left(3-\sqrt{5}\right). \label{eq:S2}
\end{eqnarray}
The amount of quantum entanglement quantified with the use of concurrence is
\begin{equation}
\mathcal{C}=\sqrt{2\left[1-\textrm{Tr}\left(\tilde{C}^{2}\right)\right]}=\frac{2}{3}
.
\end{equation}

Entanglement of the quantum histories arriving at $|a_{6}\rangle|b_{6}\rangle$ has been experimentally observed in an optical version of Hardy's setup \cite{Lundeen2009,Yokota2009} using measurement of the corresponding weak values \cite{Aharonov2002}, which can be arranged in the following weak value matrix
\begin{equation}
\hat{M}=\left(\begin{array}{cc}
\left(Q_{1,1}\right)_{w} & \left(Q_{1,2}\right)_{w}\\
\left(Q_{2,1}\right)_{w} & \left(Q_{2,2}\right)_{w}
\end{array}\right)=\left(\begin{array}{cc}
-1 & 1\\
1 & 0
\end{array}\right)
.
\end{equation}
After normalization, the weak value matrix reproduces the Schmidt coefficients \eqref{eq:S1} and \eqref{eq:S2} obtained from $\hat{C}$.
Thus, the technique of weak measurements \cite{Aharonov2014} is able to provide a bridge between the theoretical description of entangled histories in terms of propagators and their experimental registration in terms of weak values.

Quantum entanglement can be similarly quantified for the other three final
postselected states (besides $|a_{6}\rangle|b_{6}\rangle$), all of which have concurrence of $\frac{2}{3}$.

For final (postselected) state $|a_{6}\rangle|b_{5}\rangle$, the
quantum histories end with the final projector $\hat{\mathcal{P}}\left(a_{6}b_{5}\right)$
and the propagator complex coefficient matrix is (Fig.~\ref{fig:3})
\begin{equation}
\hat{C}=\left(\begin{array}{cc}
\psi_{1,1} & \psi_{1,2}\\
\psi_{2,1} & \psi_{2,2}
\end{array}\right)=\frac{1}{4}\left(\begin{array}{cc}
-\imath & -\imath\\
\imath & 0
\end{array}\right)
.
\end{equation}

\begin{figure*}[t!]
\begin{centering}
\includegraphics[width=150mm]{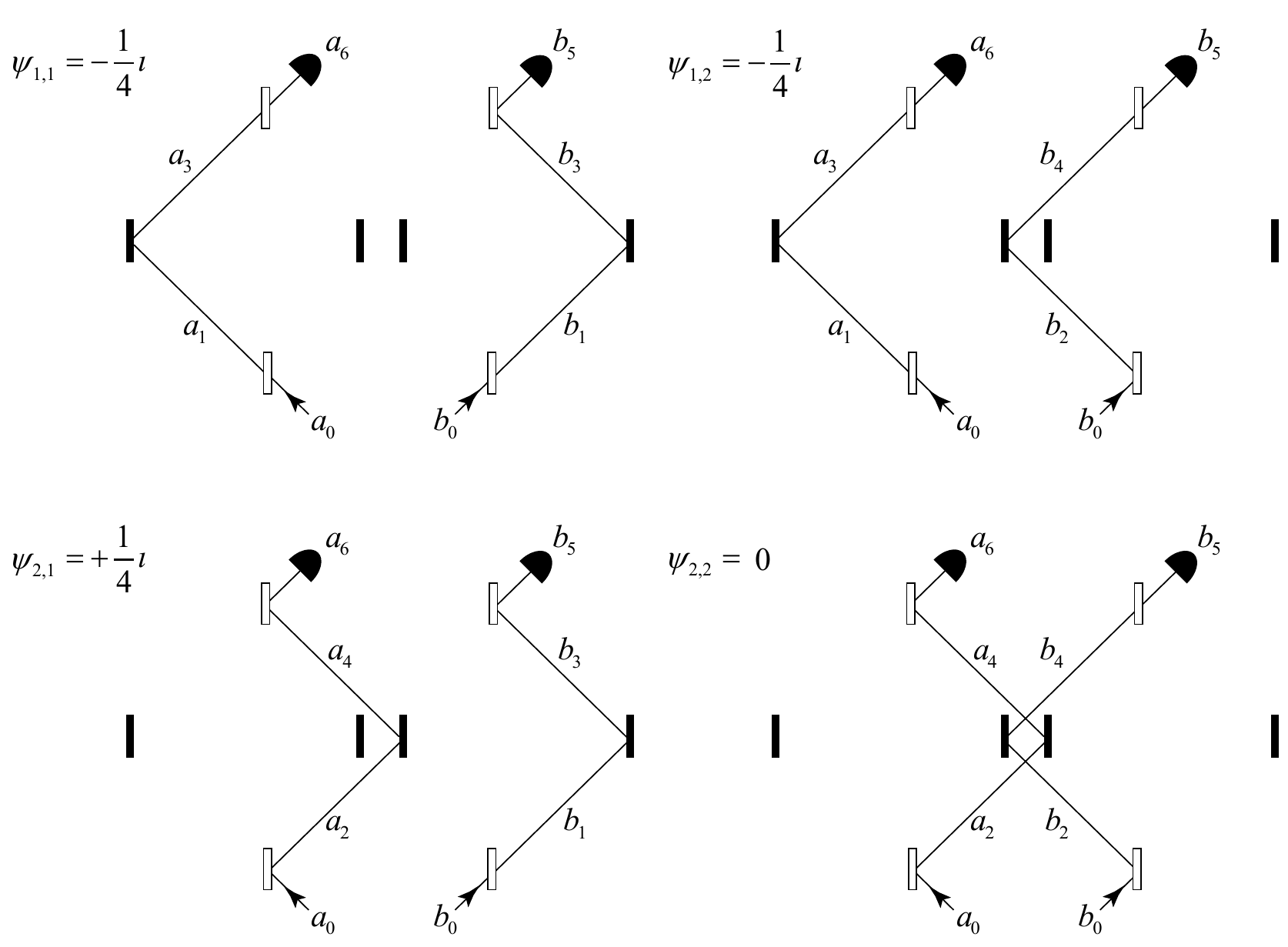}
\par\end{centering}

\caption{\label{fig:3}Illustration of the four possible composite quantum
histories, which lead to electron emerging at $a_{6}$ and positron
emerging at $b_{5}$, together with the corresponding Feynman propagators
$\psi_{\alpha,\beta}$. Due to electron-positron annihilation, one
of the histories has zero Feynman propagator, which leaves a set of
3 non-zero quantum histories in quantum entangled superposition.}
\end{figure*}

For final (postselected) state $|a_{5}\rangle|b_{6}\rangle$, the
quantum histories end with the final projector $\hat{\mathcal{P}}\left(a_{5}b_{6}\right)$
and the propagator complex coefficient matrix is (Fig.~\ref{fig:4})
\begin{equation}
\hat{C}=\left(\begin{array}{cc}
\psi_{1,1} & \psi_{1,2}\\
\psi_{2,1} & \psi_{2,2}
\end{array}\right)=\frac{1}{4}\left(\begin{array}{cc}
-\imath & \imath\\
-\imath & 0
\end{array}\right)
.
\end{equation}

\begin{figure*}[t!]
\begin{centering}
\includegraphics[width=150mm]{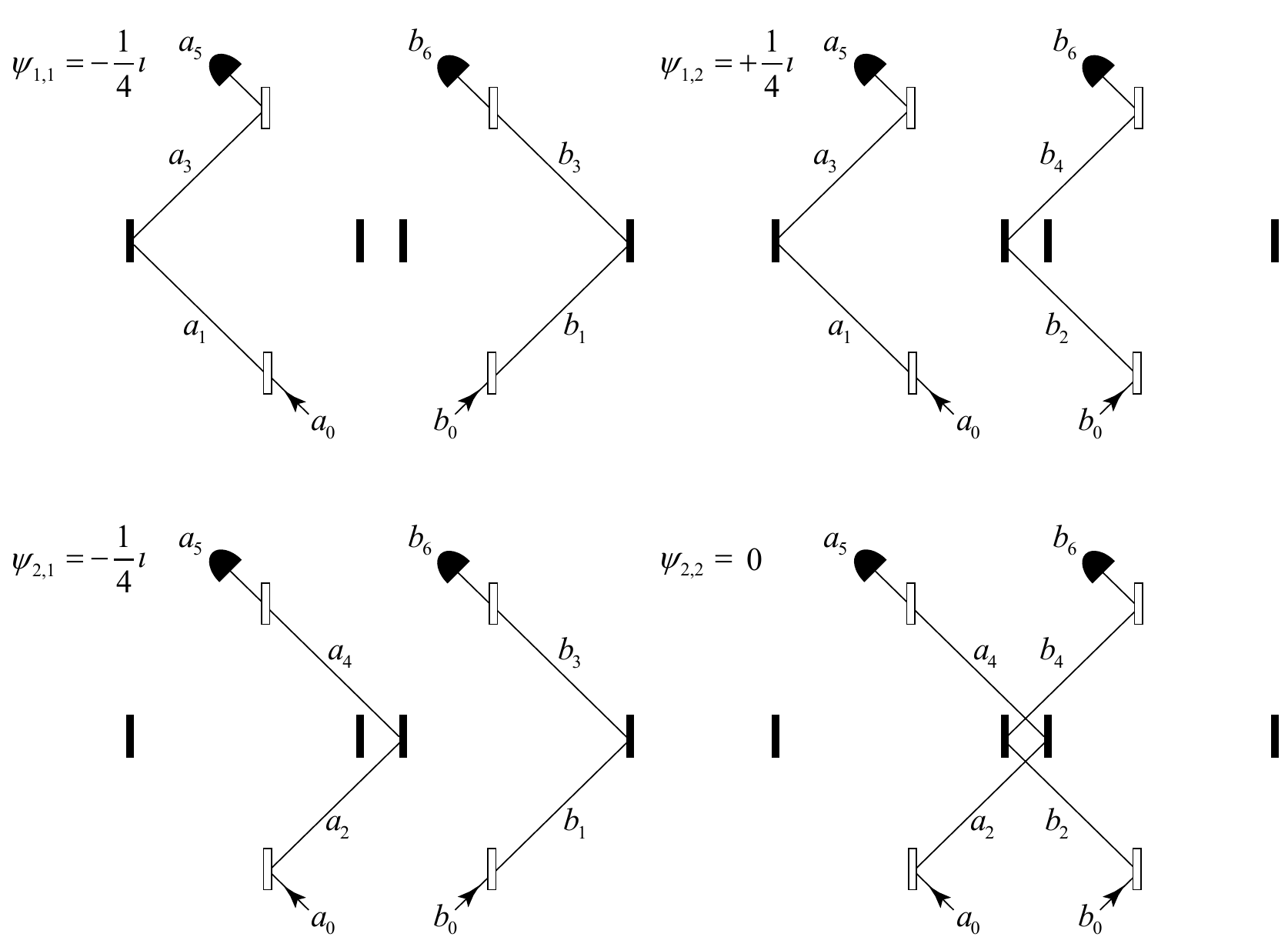}
\par\end{centering}

\caption{\label{fig:4}Illustration of the four possible composite quantum
histories, which lead to electron emerging at $a_{5}$ and positron
emerging at $b_{6}$, together with the corresponding Feynman propagators
$\psi_{\alpha,\beta}$. Due to electron-positron annihilation, one
of the histories has zero Feynman propagator, which leaves a set of
3 non-zero quantum histories in quantum entangled superposition.}
\end{figure*}

For final (postselected) state $|a_{5}\rangle|b_{5}\rangle$, the
propagator complex coefficient matrix is (Fig.~\ref{fig:5})
\begin{equation}
\hat{C}=\left(\begin{array}{cc}
\psi_{1,1} & \psi_{1,2}\\
\psi_{2,1} & \psi_{2,2}
\end{array}\right)=\frac{1}{4}\left(\begin{array}{cc}
1 & 1\\
1 & 0
\end{array}\right)
.
\end{equation}

\begin{figure*}[t!]
\begin{centering}
\includegraphics[width=150mm]{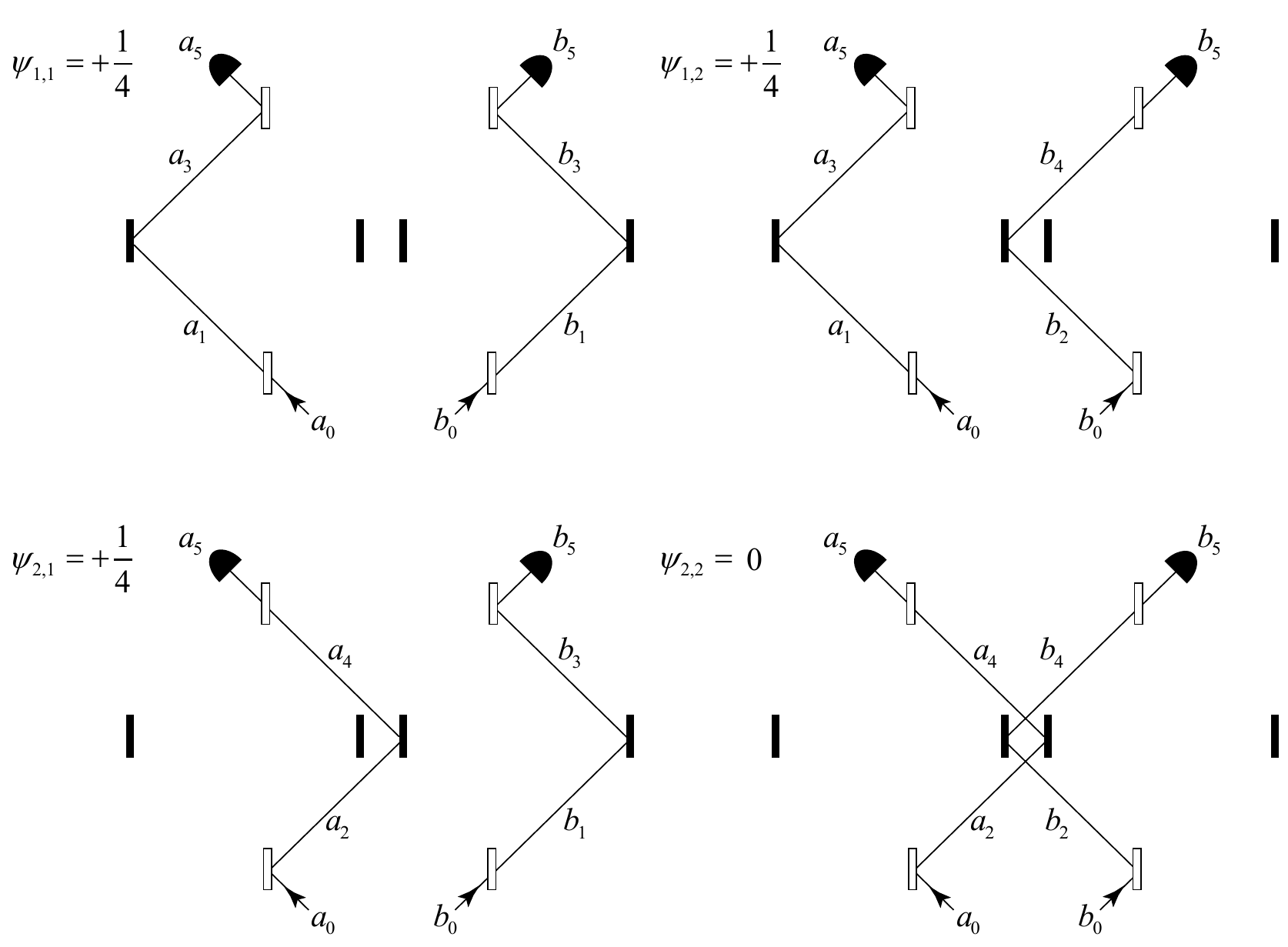}
\par\end{centering}

\caption{\label{fig:5}Illustration of the four possible composite quantum
histories, which lead to electron emerging at $a_{5}$ and positron
emerging at $b_{5}$, together with the corresponding Feynman propagators
$\psi_{\alpha,\beta}$. Due to electron-positron annihilation, one
of the histories has zero Feynman propagator, which leaves a set of
3 non-zero quantum histories in quantum entangled superposition.}
\end{figure*}

The previous four $2\times2$ propagator complex coefficient matrices
can be combined in a single $4\times4$ propagator complex coefficient
matrix in which all four pairs of final electron-positron outcomes
are present using the mapping $1\equiv1\odot5$, $2\equiv2\odot5$,
$3\equiv1\odot6$, $4\equiv2\odot6$ for the single indices $\alpha$
and $\beta$ 
\begin{equation}
\hat{C}=\left(\begin{array}{cccc}
\psi_{1,1} & \psi_{1,2} & \psi_{1,3} & \psi_{1,4}\\
\psi_{2,1} & \psi_{2,2} & \psi_{2,3} & \psi_{2,4}\\
\psi_{3,1} & \psi_{3,2} & \psi_{3,3} & \psi_{3,4}\\
\psi_{4,1} & \psi_{4,2} & \psi_{4,3} & \psi_{4,4}
\end{array}\right)=\frac{1}{4}\left(\begin{array}{cccc}
1 & 1 & -\imath & \imath\\
1 & 0 & -\imath & 0\\
-\imath & -\imath & -1 & 1\\
\imath & 0 & 1 & 0
\end{array}\right)
.
\end{equation}
Computing the concurrence from the combined $4\times4$ propagator
complex coefficient matrix also gives $\mathcal{C}=\frac{2}{3}$,
which was found for each of the individual $2\times2$ propagator
complex coefficient matrices.

The treatment of Hardy's interferometer is comparatively simple because the discretization of the setup is suggested by the particular places and times at which the particles meet the beamsplitters, mirrors and detectors. Although, formally the Hilbert space~$\mathcal{H}$ of the composite system is $7\times7=49$ dimensional, the composite quantum state vector is sparsely populated in the position basis with at most $2\times2=4$ non-zero quantum probability amplitudes at each time point. This effectively reduces the dimensionality of the Hilbert space of each subsystem to $n=2$ and greatly simplifies the computation of the Feynman propagator matrix.
For more general setups, however, it may not be evident what time steps to consider and the dimensionality of the Hilbert space may not be trivial to decide.
In such cases, a possible heuristic course of action would be to run customized numerical simulations with finer and finer coarse-graining of the position space to see if the obtained computational results converge towards a single outcome.
 
\section{\label{sec:6}Entangled quantum histories support quantum nonlocality}

Having shown that quantum systems avoid classical paradoxes by having
access to propagation along entangled quantum histories, we proceed
to show that entangled quantum histories can be used for experimental
demonstration of quantum nonlocality. For that purpose, we can use
all previous derivations for Hardy's interferometer, and supplement them
with the additional assumption that two spacelike separated agents,
Alice and Bob, are able to use their free will to either keep or remove
the final beamsplitter for the electron or the positron, respectively.
The spacelike separation is ensured by placing the final beamsplitters
far away and locating the final detectors $a_{5}$, $a_{6}$, $b_{5}$
and $b_{6}$ very close to the beamsplitters so that there is not
enough time for a light signal to inform the opposite party what their
setting of the final beamsplitter is.\\

\emph{Case 1}: Under the enforcement of spacelike separation, when
both Alice and Bob keep the final beamsplitters, from \eqref{eq:3}
we know that the final quantum state is
\begin{equation}
|\Psi_{3}^{(A_{+}B_{+})}\rangle=\frac{1}{4}\left(3\hat{a}_{5}^{\dagger}\hat{b}_{5}^{\dagger}-\imath\hat{a}_{6}^{\dagger}\hat{b}_{5}^{\dagger}-\imath\hat{a}_{5}^{\dagger}\hat{b}_{6}^{\dagger}+\hat{a}_{6}^{\dagger}\hat{b}_{6}^{\dagger}\right)|0\rangle
.
\end{equation}
The quantum probabilities for joint detections are $P\left(a_{5},b_{5}\right)=\frac{9}{16}$,
$P\left(a_{5},b_{6}\right)=\frac{1}{16}$, $P\left(a_{6},b_{5}\right)=\frac{1}{16}$
and $P\left(a_{6},b_{6}\right)=\frac{1}{16}$. The unconditional probabilities
for single particle detection are: $P(a_{5})=\frac{5}{8}$, $P(a_{6})=\frac{1}{8}$,
$P(b_{5})=\frac{5}{8}$ and $P(b_{6})=\frac{1}{8}$.\\

\emph{Case 2}: When Alice keeps the final beamsplitter for the electron,
while Bob removes the final beamsplitter for the positron, the outcomes
are obtained from updating the state \eqref{eq:a} using $\hat{b}_{3}^{\dagger}|0\rangle\to\hat{b}_{6}^{\dagger}|0\rangle$
and $\hat{b}_{4}^{\dagger}|0\rangle\to\hat{b}_{5}^{\dagger}|0\rangle$
as follows
\begin{equation}
|\Psi_{3}^{(A_{+}B_{-})}\rangle=-\frac{1}{\sqrt{8}}\left[2\imath\hat{a}_{5}^{\dagger}\hat{b}_{6}^{\dagger}-\hat{a}_{5}^{\dagger}\hat{b}_{5}^{\dagger}+\imath\hat{a}_{6}^{\dagger}\hat{b}_{5}^{\dagger}\right]|0\rangle
.
\end{equation}
In this case, the quantum probabilities for joint detections are $P(a_{5},b_{5})=\frac{1}{8}$,
$P(a_{5},b_{6})=\frac{1}{2}$, $P(a_{6},b_{5})=\frac{1}{8}$ and $P(a_{6},b_{6})=0$.
The unconditional probabilities for single particle detection are:
$P(a_{5})=\frac{5}{8}$, $P(a_{6})=\frac{1}{8}$, $P(b_{5})=\frac{1}{4}$
and $P(b_{6})=\frac{1}{2}$.\\

\emph{Case 3}: Similarly, when Alice removes the final beamsplitter
for the electron, while Bob keeps the final beamsplitter for the positron,
the outcomes are obtained from updating the state \eqref{eq:b} using
$\hat{a}_{3}^{\dagger}|0\rangle\to\hat{a}_{6}^{\dagger}|0\rangle$
and $\hat{a}_{4}^{\dagger}|0\rangle\to\hat{a}_{5}^{\dagger}|0\rangle$
as follows
\begin{equation}
|\Psi_{3}^{(A_{-}B_{+})}\rangle=-\frac{1}{\sqrt{8}}\left[2\imath\hat{a}_{6}^{\dagger}\hat{b}_{5}^{\dagger}-\hat{a}_{5}^{\dagger}\hat{b}_{5}^{\dagger}+\imath\hat{a}_{5}^{\dagger}\hat{b}_{6}^{\dagger}\right]|0\rangle
.
\end{equation}
In this case, the quantum probabilities for joint detections are $P(a_{5},b_{5})=\frac{1}{8}$,
$P(a_{5},b_{6})=\frac{1}{8}$, $P(a_{6},b_{5})=\frac{1}{2}$ and $P(a_{6},b_{6})=0$.
The unconditional probabilities for single particle detection are:
$P(a_{5})=\frac{1}{4}$, $P(a_{6})=\frac{1}{2}$, $P(b_{5})=\frac{5}{8}$
and $P(b_{6})=\frac{1}{8}$.\\

\emph{Case 4}: Lastly, we will need to compute the probabilities for
joint detection when both Alice and Bob remove their final beamsplitters.
The outcomes are obtained from updating the state \eqref{eq:2} using
$\hat{a}_{3}^{\dagger}|0\rangle\to\hat{a}_{6}^{\dagger}|0\rangle$,
$\hat{a}_{4}^{\dagger}|0\rangle\to\hat{a}_{5}^{\dagger}|0\rangle$,
$\hat{b}_{3}^{\dagger}|0\rangle\to\hat{b}_{6}^{\dagger}|0\rangle$
and $\hat{b}_{4}^{\dagger}|0\rangle\to\hat{b}_{5}^{\dagger}|0\rangle$
as follows
\begin{equation}
|\Psi_{3}^{(A_{-}B_{-})}\rangle=-\frac{1}{2}\left(\hat{a}_{6}^{\dagger}\hat{b}_{6}^{\dagger}+\imath\hat{a}_{6}^{\dagger}\hat{b}_{5}^{\dagger}+\imath\hat{a}_{5}^{\dagger}\hat{b}_{6}^{\dagger}\right)|0\rangle
.
\end{equation}
In this case, the quantum probabilities for joint detections are $P(a_{5},b_{5})=0$,
$P(a_{5},b_{6})=\frac{1}{4}$, $P(a_{6},b_{5})=\frac{1}{4}$ and $P(a_{6},b_{6})=\frac{1}{4}$.
The unconditional probabilities for single particle detection are:
$P(a_{5})=\frac{1}{4}$, $P(a_{6})=\frac{1}{2}$, $P(b_{5})=\frac{1}{4}$
and $P(b_{6})=\frac{1}{2}$.\\

Now, if one conjectures that the electron and the positron could have
produced their arrival at their corresponding particle detectors $a_{6}$
and $b_{6}$ without taking into consideration the setting of the
distant spacelike separated beamsplitter, a contradiction will occur
as follows: Consider the most general local strategy that the electron
and positron could employ by taking into consideration only the local
setting of the corresponding final beamsplitter together with the
past interferometer arm from which they come in order to accommodate
the correct $\frac{1}{4}$ probability of electron-positron annihilation.
There will be four non-negative probability weights $x_{1}^{\pm},x_{2}^{\pm}$
for the electron to go to detector $a_{6}$ depending on the interferometer
arm $1,2$ and the beamsplitter setting $\pm$. Similarly, there will
be four non-negative probability weights $y_{1}^{\pm},y_{2}^{\pm}$
for the positron to go to detector $b6$. The joint probability of
arrival at $a_{6}$ and $b_{6}$ for the local model should match
the quantum probability $P(a_{6},b_{6})$ for all four possible beamsplitter
settings, namely 
\begin{align}
\frac{1}{4}\left(x_{1}^{+}y_{1}^{+}+x_{1}^{+}y_{2}^{+}+x_{2}^{+}y_{1}^{+}\right) & \neq0 , \label{eq:l1}\\
\frac{1}{4}\left(x_{1}^{+}y_{1}^{-}+x_{1}^{+}y_{2}^{-}+x_{2}^{+}y_{1}^{-}\right) & =0 , \label{eq:l2}\\
\frac{1}{4}\left(x_{1}^{-}y_{1}^{+}+x_{1}^{-}y_{2}^{+}+x_{2}^{-}y_{1}^{+}\right) & =0 , \label{eq:l3}\\
\frac{1}{4}\left(x_{1}^{-}y_{1}^{-}+x_{1}^{-}y_{2}^{-}+x_{2}^{-}y_{1}^{-}\right) & \neq0 .\label{eq:l4}
\end{align}
Further, since the unconditional quantum probabilities are $P(a_{6})\neq0$
and $P(b_{6})\neq0$ for all four possible beamsplitter settings,
we also should have
\begin{align}
x_{1}^{+}+x_{2}^{+} & \neq0 ,\label{eq:l5}\\
x_{1}^{-}+x_{2}^{-} & \neq0 ,\label{eq:l6}\\
y_{1}^{+}+y_{2}^{+} & \neq0 ,\label{eq:l7}\\
y_{1}^{-}+y_{2}^{-} & \neq0 .\label{eq:l8}
\end{align}
To show that the system of equations \eqref{eq:l1}-\eqref{eq:l8}
does not have a solution, one can derive a contradiction in multiple
ways. One of the shortest derivations is to substitute \eqref{eq:l8}
in \eqref{eq:l2} to get $x_{1}^{+}=0$. Similarly, substitute \eqref{eq:l6}
in \eqref{eq:l3} to get $y_{1}^{+}=0$. Then, the joint substitution
of $x_{1}^{+}=y_{1}^{+}=0$ in \eqref{eq:l1} leads to contradiction
$0\neq0$, which implies that there exists no local hidden variable
model that is able to reproduce the experimental results from
Hardy's interferometer. Because the mathematical proof does not use the actual
numerical values of the probabilities, but only cares to distinguish
between zero and non-zero probabilities, the demonstration of quantum
nonlocality by this approach is often referred to as Bell's theorem
without inequalities \cite{Hardy1993,Cabello2001}.

Quantum nonlocality manifested in the observable correlations of spacelike separated quantum measurements
cannot be used for superluminal communication because the unconditional
probabilities for each local outcome remain unaffected by the distant choice
of measurement basis \cite{Peres2004}. Indeed, pairwise comparison
of the four cases in Hardy's interferometer verifies that the unconditional
quantum probabilities for different measurement outcomes $a_5$, $a_6$, $b_5$ or $b_6$ are affected
only by the local setting of the final beamsplitter, but remain unchanged
by the setting of the other spacelike separated final beamsplitter.

\section{Conclusion}

In this work, we have elaborated on three related, but mathematically distinct objects
that describe the notion of \emph{quantum history} in different contexts.
Firstly,
we have clarified the fact that in the context of the history Hilbert space $\breve{\mathcal{H}}$,
the bipartite quantum histories are constructed as mutually orthogonal projection operators $\textrm{Tr}[\mathcal{\hat{Q}}_{\alpha,\beta}\cdot\mathcal{\hat{Q}}_{\alpha^{\prime},\beta^{\prime}}]=\delta_{\alpha\alpha^{\prime}}\delta_{\beta\beta^{\prime}}$.
Employing only local projection operators for each subsystem at each time point ensures that the resulting bipartite quantum history projection operators, which span the history Hilbert space, are separable (not entangled), $\mathcal{\hat{Q}}_{\alpha,\beta}=\mathcal{\hat{Q}}_{\alpha}\otimes\mathcal{\hat{Q}}_{\beta}$, and hence suitable for Schmidt decomposition of the quantum history state vector $|\Psi)$.
Secondly,
we have revealed that quantum entanglement is generated by the quantum
Hamiltonian~$\hat{H}$ acting on the composite quantum state~$|\Psi\rangle$
in the standard Hilbert space $\mathcal{H}$. If the quantum Hamiltonian~$\hat{H}$ is explicitly given, it can be used for the calculation
of time development operators~$\hat{\mathcal{T}}$ that replace the
time tensor product symbols~$\odot$ in~$\mathcal{\hat{Q}}_{\alpha,\beta}$ for
the formation of corresponding quantum history chain operators~$\hat{K}_{\alpha,\beta}$ in standard Hilbert space.
Due to the presence of time development operators, the chain operators~$\hat{K}_{\alpha,\beta}$ of different quantum histories are no longer guaranteed to be mutually orthogonal.
Feynman's sum-over-histories formulation is extremely helpful in situations when the time development operators are already known, because it is not necessary to reconstruct the quantum Hamiltonian in order to compute the Feynman propagators.
Thirdly,
by fixing the final postselected state of the composite system, say $|\Psi_{k+1}\rangle$, we have shown
that one computes the quantum history Feynman propagator $\psi_{\alpha,\beta}$ using the
inner product involving the corresponding chain operator $\langle\Psi_{k+1}|\hat{K}_{\alpha,\beta}|\Psi_{0}\rangle$
and then uses the information obtained from all quantum history propagators
to determine the possible entanglement of the quantum histories.

There are several contributions in our theoretical approach including the realization that
neither entangled projectors in the history Hilbert space,
nor inner products between pairs of quantum history chain operators in the standard Hilbert space
are required for quantification of entanglement of bipartite quantum histories.
Instead, we have demonstrated that the singular value decomposition of the propagator
complex coefficient matrix $\hat{C}$ contains all the information necessary for answering the question
whether a complete set of bipartite quantum histories is entangled or not.
In fact, the standard Schmidt decomposition of a bipartite state vector $|\Psi(t)\rangle$ at a
single time point $t$ could be viewed as a special (trivial) case of Schmidt
decomposition of quantum history state vector $|\Psi)$ that possesses only a single time point.

From the Schmidt coefficients obtained in the decomposition of the
propagator complex coefficient matrix $\hat{C}$, one is able to compute
a number of entanglement measures, including entanglement
entropy $\mathcal{S}$, entanglement robustness $\mathcal{R}$ and
concurrence $\mathcal{C}$. Although in Hardy's overlapping interferometers, it is relatively easy to perform the singular value
decomposition for the extraction of the Schmidt coefficients, this
task becomes very expensive computationally for large history Hilbert
spaces. Fortunately, for the quantitative evaluation of concurrence
there is a computational workaround proposed by Gudder \cite{Gudder2020b,Georgiev2022},
namely, rather than performing singular value decomposition one can
directly compute the trace of the matrix $\tilde{C}^{2}$ and plug it into~\eqref{eq:concurrence}.
This establishes the computational ease of concurrence over other
entanglement measures for witnessing entanglement of quantum histories.

Entanglement of quantum histories is a robust prediction of the standard quantum formalism, which holds a great explanatory power with regard to occurrence of classically inexplicable experimental results in quantum foundations. Previous works on ``entangled histories'' have explored the possible non-factorizability of quantum histories in time for single quantum systems and have shown how entanglement-in-time can be utilized for experimental testing and verification of \emph{quantum contextuality} \cite{Cotler2016,Cotler2017b,Dong2017,Nowakowski2018,Pan2019}. Motivated by the additional opportunities provided by quantum system composition, here we have investigated the possible entanglement of bipartite quantum histories and quantified the amount of quantum entanglement that can be utilized for experimental testing and verification of \emph{quantum nonlocality} between spacelike separated agents.

Because our approach makes an extensive use of Feynman propagators, it is well-suited for the analysis and design of optical applications comprised of e.g. optical fibers connected to quantum controllers. Straightforward description of such optical setups is aided by the repetitive use of known time development operators with substantial effective dimensional reduction of the constructed history Hilbert space. The accessibility of different Feynman propagators through measurable weak values also provides potentially useful means for calibration of engineered optical quantum devices and direct verification of their physical mechanism of action.

The prospects for future research include possible extension of the quantum history formalism to multipartite quantum systems and exploration of different generalizations of Schmidt decomposition that go beyond simple consideration of the set of all bipartitions. Another interesting research avenue would be to consider the constraints on entanglement of quantum histories imposed by the presence of indistinguishable particles.

\section*{Acknowledgments}

We would like to thank the anonymous reviewers for constructive comments. E.C. was supported by the Israeli Innovation Authority under Projects No. 70002 and No. 73795, by the Pazy Foundation, by the Israeli Ministry of Science and Technology, and by the Quantum Science and Technology Program of the Israeli Council of Higher Education.

\bibliographystyle{apsrev4-2}
\bibliography{references}

\end{document}